\pdfoutput=1
\documentclass[11pt,a4paper]{article}
\usepackage[utf8]{inputenc}
\usepackage{lmodern} 
\usepackage{palatino}
\usepackage[bookmarks,colorlinks,breaklinks]{hyperref}
\hypersetup{urlcolor=blue, colorlinks=true, citecolor=green!50!black, linkcolor=blue}

\usepackage[T1]{fontenc}
\usepackage{amsmath}
\usepackage{amsfonts}
\usepackage{amssymb}
\usepackage{amsthm}
\usepackage{thmtools}
\usepackage{cleveref}
\usepackage{geometry}
\usepackage{enumitem}
\usepackage{xcolor}
\usepackage{stmaryrd}
\declaretheorem[numberwithin=section,refname={Theorem,Theorems},Refname={Theorem,Theorems}]{theorem}
\declaretheorem[numberlike=theorem]{lemma}

\declaretheorem[numberlike=theorem]{claim}

\DeclareMathOperator*{\argmin}{arg\,min}

\def\poly{\operatorname{poly}}
\def\polylog{\operatorname{polylog}}
\def\head{\operatorname{head}}
\def\tail{\operatorname{tail}}
\def\tind{\mathcal{T}_{\mathsf{ind}}}
\def\trank{\mathcal{T}_{\mathsf{rank}}}
\def\rank{\mathsf{rank}}

\renewcommand{\tilde}{\widetilde}
\renewcommand{\hat}{\widehat}
\renewcommand{\bar}{\overline}

\makeatletter
\def\namedlabel#1#2{\begingroup
   \def\@currentlabel{#2}%
   \label{#1}\endgroup
}
\makeatother

\pdfoutput=1

\title{Subquadratic Weighted Matroid Intersection Under Rank Oracles\thanks{
To appear in the \emph{33rd International Symposium on Algorithms
and Computation (ISAAC 2022)}.}}

\author{Ta-Wei Tu\thanks{Department of Computer Science and Information
  Engineering, National Taiwan University, Taipei, Taiwan.
  Email:
  \href{mailto:tu.da.wei@gmail.com}{tu.da.wei@gmail.com}}}
\date{}

\begin{document}

\maketitle

\begin{abstract}
  Given two matroids $\mathcal{M}_1 = (V, \mathcal{I}_1)$ and
  $\mathcal{M}_2 = (V, \mathcal{I}_2)$ over an $n$-element integer-weighted
  ground set $V$, the weighted matroid intersection problem aims to find a
  common independent set $S^{*} \in \mathcal{I}_1 \cap \mathcal{I}_2$ maximizing
  the weight of $S^{*}$.
  In this paper, we present a simple deterministic algorithm for weighted
  matroid intersection using
  $\tilde{O}(nr^{3/4}\log{W})$ rank queries,
  where $r$ is the size of the largest intersection of $\mathcal{M}_1$ and
  $\mathcal{M}_2$ and $W$ is the maximum weight.
  This improves upon the best previously known $\tilde{O}(nr\log{W})$
  algorithm
  given by Lee, Sidford, and Wong [FOCS'15], and is the first subquadratic
  algorithm for polynomially-bounded weights under the standard
  independence or rank oracle models.
  The main contribution of this paper is an efficient algorithm that computes
  shortest-path trees in weighted exchange graphs.
\end{abstract}

\section{Introduction}\label{sec:intro}

\paragraph{Matroid Intersection.}
A matroid is an abstract structure that models the notion of independence on a
given ground set $V$.
In particular, a subset $S \subseteq V$ is either \emph{independent} or
\emph{dependent}, such that the family of independent sets is well-structured
(see \Cref{sec:prelim} for a complete definition).
Matroids model many fundamental combinatorial objects, and examples of
independent sets of a matroid include acyclic subgraphs of an undirected graph
and linearly independent rows of a matrix.
One of the most important optimization problems related to matroids is
\emph{matroid intersection}:
Given two matroids, we would like to find a set with the largest cardinality
that is independent in both matroids.
Similarly, in the weighted case, each element in the ground set is associated
with an integer weight, and the weighted matroid intersection problem is to
find the maximum-weight common independent set.
These problems have been extensively studied in the past since they capture
many combinatorial optimization problems such as bipartite
matching and colorful spanning trees.

\paragraph{Oracle Model.}
Since we are dealing with general matroids without additional constraints,
we have to specify a way of reading the description of the two matroids.
One way is to express them directly by reading the truth table of independence.
However, that would require an exponentially-sized input.
Instead, we are given oracle access to the matroids,
which gives us information about a queried set $S \subseteq V$.
Standard oracles include the \emph{independence oracle}, which returns
whether $S$ is independent in $O(\tind)$ time, and
the \emph{rank oracle}, which returns the rank, i.e., the size of the
largest independent subset, of $S$ in $O(\trank)$ time.
In this paper, we focus on the stronger rank oracle model.

\paragraph{Prior Work.}
Polynomial-time algorithms for both the weighted and
unweighted matroid intersection problems have long been designed and
improved.
For the unweighted case,
Edmonds~\cite{edmonds1968,edmonds1970,edmonds1979},
Lawler~\cite{lawler1975},
and also Aigner and Dowling~\cite{aigner1971} gave algorithms that run in
$O(nr^2 \cdot \tind)$ time.
Here, $n$ denotes the number of elements in $V$ and $r$ denotes the size of
the largest intersection of the two matroids.
Cunningham~\cite{cunningham1986} obtained an $O(nr^{3/2} \cdot \tind)$
algorithm using the ``blocking-flow'' idea.
Lee, Sidford, and Wong~\cite{lee2015} gave quadratic algorithms using
the cutting-plane method,
running in $\tilde{O}(nr \cdot \trank + n^3)$
and $\tilde{O}(n^2 \cdot \tind + n^3)$ times\footnote{
For function $f(n)$, $\tilde{O}(f(n))$ denotes $O(f(n)\polylog{f(n)})$.}
, respectively.
This also gives rise to the ``quadratic barrier'' of matroid intersection: most
previous algorithms involve building exchange graphs that contain $\Theta(nr)$
edges explicitly and therefore cannot go beyond quadratic time.
Chakrabarty, Lee, Sidford, Singla, and Wong~\cite{chakrabarty2019}
were the first to partially break the barrier.
They obtained a
$(1 - \epsilon)$-approximation algorithm running in
$\tilde{O}(n^{3/2}/\epsilon^{3/2} \cdot \tind)$
time and an
$\tilde{O}(n\sqrt{r} \cdot \trank)$ exact algorithm.
One of the major components of Chakrabarty et al.'s improvements is to show that
edges in exchange graphs can be efficiently discovered using binary search
(this was discovered independently by Nguy\~{\^{e}}n~\cite{nguyen2019}).
This technique also allows them to obtain improved $\tilde{O}(nr \cdot \tind)$
exact algorithms.
Combining the approximation algorithm and a faster augmenting-path algorithm,
Blikstad, van den Brand, Mukhopadhyay, and Nanongkai~\cite{blikstad2021}
broke the quadratic barrier completely by
giving an $\tilde{O}(n^{9/5} \cdot \tind)$ exact algorithm.
This result was later optimized by Blikstad~\cite{blikstad2021-2} to
$\tilde{O}(nr^{3/4} \cdot \tind)$ by improving the approximation
algorithm to run in $\tilde{O}(n\sqrt{r}/\epsilon \cdot \tind)$ time.

For the weighted case, the blocking flow idea does not seem to apply anymore.
Frank~\cite{frank1981} obtained an $O(nr^2 \cdot \tind)$
algorithm by characterizing the optimality of a common independent set using
weight splitting.
Fujishige and Zhang~\cite{fujishige1995} improved the running time to
$\tilde{O}(nr^{3/2}\log{W} \cdot \tind)$ by solving a more
general \emph{independent assignment} problem using a scaling framework.
The same bound was achieved by Shigeno and Iwata~\cite{shigeno1995} and also
by Gabow and Xu~\cite{gabow1996}.
Lee, Sidford, and Wong's~\cite{lee2015} algorithms work for the weighted case as
well, albeit with an extra factor of $\polylog{W}$, in
$\tilde{O}(n^2\log{W} \cdot \tind + n^3\polylog{W})$ and
$\tilde{O}(nr\log{W} \cdot \trank + n^3\polylog{W})$ times.
Huang, Kakimura, and Kamiyama~\cite{huang2016} obtained a generic framework that
transforms any algorithm that solves the unweighted case into one that solves
the weighted case with an extra $O(W)$ factor.
Plugging in the state-of-the-art algorithms of~\cite{chakrabarty2019} and
\cite{blikstad2021-2}, we get
$\tilde{O}(n\sqrt{r} \cdot W \cdot \trank)$ and
$\tilde{O}(nr^{3/4} \cdot W \cdot \tind)$ algorithms.
Chekuri and Quanrud~\cite{chekuri2016} also gave an
$\tilde{O}(n^2/\epsilon^2 \cdot \tind)$ approximation
algorithm which, according to~\cite{blikstad2021},
can be improved to subquadratic by applying more recent techniques.
A similar $\tilde{O}(nr^{3/2}/\epsilon \cdot \tind)$
approximation algorithm was obtained independently by
Huang et al.~\cite{huang2016}.

\paragraph{Our Result.}
The question of whether weighted matroid intersection can be solved in
subquadratic time with polylogarithmic dependence on $W$ under either oracle
model remained open.
We obtain the first subquadratic algorithm for exact weighted matroid
intersection under rank oracles.
The formal statement of \Cref{thm:result} is presented as
\Cref{thm:weighted_matroid_intersection} in \Cref{sec:prelim}.

\begin{theorem}
  Weighted matroid intersection can be solved in
  $\tilde{O}(nr^{3/4}\log{W} \cdot \trank)$ time.
  \label{thm:result}
\end{theorem}

Our algorithm relies on the framework of
Fujishige-Zhang~\cite{fujishige1995} and Shigeno-Iwata~\cite{shigeno1995},
where they first obtain an approximate solution by adjusting weights
of some elements (similar to the ``auction`` algorithms for bipartite
matching~\cite{orlin1992}) and then refine it by augmenting the solution
iteratively.

We obtain efficient algorithms for these two phases, leading to the
final subquadratic algorithm.

\section{Preliminaries}\label{sec:prelim}

\paragraph{Notation.}
For a set $S$, let $|S|$ denote the cardinality and $2^{S}$
the power set of $S$.
Let $S \setminus R$ consist of elements of $S$ which are not in $R$.
Let $e = (u, v, w)$ denote a weighted directed edge directing from $u$ to $v$
with weight $w = w(e)$ and $(u, v)$ be its unweighted counterpart.
Let $\head(e) = v$ and $\tail(e) = u$.
For an edge set $E$, let
$\head(E) = \{\head(e) \mid e \in E\}$ and
$\tail(E) = \{\tail(e) \mid e \in E\}$.
For functions $f, g$ mapping from a set $V$ to $\mathbb{R}$,
let $f + g$, $f - g$, and $f + c$ for $c \in \mathbb{R}$ denote
functions from $V$ to $\mathbb{R}$ with
$(f + g)(x) = f(x) + g(x)$,
$(f - g)(x) = f(x) - g(x)$, and
$(f + c)(x) = f(x) + c$ for each $x \in V$.
We often abuse notation and use $f$ to denote the function from
$2^V$ to $\mathbb{R}$ with $f(S) = \sum_{x \in S}f(x)$ for
each $S \subseteq V$.

\paragraph{Matroid.}
Let $V$ be a finite set and $w: V \to \mathbb{Z}$ be a given weight function.
For $S \subseteq V$, let $\bar{S} = V \setminus S$.
Let $n = |V|$ and $W = \max_{x \in V}|w(x)|$.
An ordered pair $\mathcal{M} = (V, \mathcal{I})$ with \emph{ground set} $V$ and
a non-empty family $\emptyset \in \mathcal{I} \subseteq 2^V$ is a \emph{matroid}
if

\begin{description}
  \item[M1.\namedlabel{item:downward_closure}{M1}] for each $S \in \mathcal{I}$
    and $R \subseteq S$, it holds that $R \in \mathcal{I}$, and
  \item[M2.] for each $R, S \in \mathcal{I}$ with $|R| < |S|$,
    there exists an $x \in S \setminus R$ such that
    $R \cup \{x\} \in \mathcal{I}$.
\end{description}

Sets in $\mathcal{I}$ are \emph{independent}; sets not in $\mathcal{I}$ are
\emph{dependent}.
A \emph{basis} is a maximal independent set.
A \emph{circuit} is a minimal dependent set.
It is well-known from the definition of matroid that all bases are of the same
cardinality.
For an independent set $S$ and $x \not\in S$, $S \cup \{x\}$ contains at most
one circuit $C$ and if it does, then $x \in C$
(see \cite[Lemma 1.3.3]{price2015}).
The \emph{rank} of $S \subseteq V$, denoted by $\rank(S)$, is the size of the
largest $S^\prime \subseteq S$ such that $S^\prime \in \mathcal{I}$.
The rank of $\mathcal{M}$ is the rank of $V$, i.e., the size of the bases of
$\mathcal{M}$.
Given two matroids $\mathcal{M}_1 = (V, \mathcal{I}_1)$ and
$\mathcal{M}_2 = (V, \mathcal{I}_2)$ over the same ground set,
the \emph{weighted matroid intersection}
problem is to find an $S^{*} \in \mathcal{I}_1 \cap \mathcal{I}_2$ maximizing
$w(S^{*})$.
Let $r = \max_{S \in \mathcal{I}_1 \cap \mathcal{I}_2}|S|$.
In this paper, the two matroids are accessed through \emph{rank oracles}, one
for each matroid.
Specifically, let $\rank_1(\cdot)$ and $\rank_2(\cdot)$ denote
the rank functions of $\mathcal{M}_1$ and $\mathcal{M}_2$, respectively.
We assume that given pointers to a linked list containing elements of $S$
(see, e.g.,~\cite{chakrabarty2017}), the rank oracles compute
$\rank_1(S)$ and $\rank_2(S)$ in $O(\trank)$ time.
With the $O(n\sqrt{r}\log{n} \cdot \trank)$ unweighted
matroid intersection algorithm of Chakrabarty et al.~\cite{chakrabarty2019},
we also assume that $\mathcal{M}_1$ and $\mathcal{M}_2$ are of the same rank
and share a common
basis $S^{(0)}$ of size $r$ by adjusting the given rank oracles
properly.\footnote{We can compute $r$ via the unweighted matroid intersection
algorithm and regard all sets of size greater than $r$ as dependent.}
By adding $r$ zero-weight elements to $V$, we may also assume that each common
independent set $S \in \mathcal{I}_1 \cap \mathcal{I}_2$ is contained in a
common basis of the same weight.\footnote{In particular, let
$Z = \{z_1, \ldots, z_r\}$ be the set of newly added zero-weight elements.
For each $i \in \{1, 2\}$, instead of working with $\mathcal{M}_i$, we now
work with $\tilde{\mathcal{M}}_i = (V \cup Z, \tilde{\mathcal{I}}_i)$
such that for each $\tilde{S} \subseteq V \cup Z$,
$\tilde{S} \in \tilde{\mathcal{I}}_i$ if and only if
$|\tilde{S}| \leq r$ and $\tilde{S} \setminus Z \in \mathcal{I}_i$.
This change is reflected in the new rank function
$\tilde{\rank}_i(\tilde{S}) =
\min(\rank_i(\tilde{S} \setminus Z) + |\tilde{S} \cap Z|, r)$,
which can be implemented via the given oracle $\rank_i$.}
Therefore, it suffices to find a common basis $S^{*}$ maximizing $w(S^{*})$.
Note that elements with negative weights can be safely discarded from $V$.

\paragraph{Weight-Splitting.}
For weight function $f: V \to \mathbb{R}$, a basis $S$ of matroid
$\mathcal{M}$ is \emph{$f$-maximum} if $f(S) \geq f(R)$ holds for each basis $R$
of $\mathcal{M}$.
Let $w^{\epsilon} = (w^{\epsilon}_1, w^{\epsilon}_2)$ with
$w^{\epsilon}_i: V \to \mathbb{R}$ being a weight function
for each $i \in \{1, 2\}$.
Let $w^{\epsilon}(x) = w^{\epsilon}_1(x) + w^{\epsilon}_2(x)$.
We say that $w^{\epsilon}$ is an \emph{$\epsilon$-splitting}
(see, e.g.,~\cite{shigeno1995}) of $w$ with $\epsilon > 0$ if
$w(x) \leq w^{\epsilon}(x) \leq w(x) + \epsilon$ holds for each $x \in V$.
If $w^{\epsilon}$ is an $\epsilon$-splitting of $w$ and $S_i$ is a
$w^{\epsilon}_i$-maximum basis of $\mathcal{M}_i$ for each $i \in \{1, 2\}$,
then we call $(w^{\epsilon}, S)$ with $S = (S_1, S_2)$ an
\emph{$\epsilon$-partial-solution} of $w$.
Note that by~\ref{item:downward_closure}, $S_1 \cap S_2$ is a common independent
set.
If $S_1 = S_2$, then $(w^{\epsilon}, S)$ is an \emph{$\epsilon$-solution} of $w$.
In this case, we may abuse notation and refer to $S_1$ as simply $S$.

\paragraph{Matroid Algorithms.}
The unweighted version of the following lemma was shown in
\cite{chakrabarty2019} (it is also mentioned in~\cite{nguyen2019}),
and it was extended to the weighted case implicitly in
\cite{blikstad2021}.

\begin{lemma}[{\cite[Lemma 13]{chakrabarty2019}},
  {\cite{nguyen2019}}, and {\cite{blikstad2021}}]
  For $i \in \{1, 2\}$,
  given $S \in \mathcal{I}_i$, $B \subseteq S$ (respectively,
  $B \subseteq \bar{S}$), $x \in \bar{S}$ (respectively, $x \in S$), and
  weight function $f: V \to \mathbb{R}$, it takes
  $O(\log{|B|} \cdot \trank)$ time to either obtain a $b \in B$
  minimizing/maximizing $f(b)$ such that
  $(S \setminus \{b\}) \cup \{x\} \in \mathcal{I}_i$ (respectively,
  $(S \setminus \{x\}) \cup \{b\} \in \mathcal{I}_i$) or
  report that such an element does not exist in $B$.
  \label{lemma:find_exchange}
\end{lemma}

The main idea of \Cref{lemma:find_exchange} is to perform binary search on $B$
ordered by $f$.
Throughout this paper, we will maintain such an ordered set in a balanced
binary search tree where each element holds pointers to its successor and
predecessor and each node holds pointers to the first and the last elements
in its corresponding subtree.
This allows us to perform binary search on the tree and obtain pointers to the
linked list containing elements in a consecutive range efficiently.

The following greedy algorithm for finding a maximum-weight basis is folklore.

\begin{lemma}[See, e.g., \cite{edmonds1971}]
  It takes $O(n\log{n} + n\trank)$ time to obtain a $f$-maximum basis
  $S$ of a given matroid $\mathcal{M}$ and weight function
  $f: V \to \mathbb{R}$.
  \label{lemma:greedy}
\end{lemma}

\subsection{The Framework}

The core of our algorithm is the following subroutine.

\begin{theorem}
  Given a $2\epsilon$-solution $(w^{2\epsilon}, S^\prime)$ of $w$,
  it takes $O(nr^{3/4}\log{n} \cdot \trank)$ time to obtain
  an $\epsilon$-solution $(w^{\epsilon}, S)$.
  \label{thm:main}
\end{theorem}

With \Cref{thm:main}, the weighted matroid intersection algorithm follows
from the standard weight-scaling framework
(see, e.g, \cite{fujishige1995,shigeno1995}).
Recall that our goal is to find a maximum-weight common basis.

\begin{theorem}[Weighted Matroid Intersection]
  Given two matroids $\mathcal{M}_1 = (V, \mathcal{I}_1)$ and
  $\mathcal{M}_2 = (V, \mathcal{I}_2)$, it takes
  $O(nr^{3/4}\log{n}\log{(rW)} \cdot \trank)$ time to
  obtain an $S^{*} \in \mathcal{I}_1 \cap \mathcal{I}_2$ maximizing $w(S^{*})$.
  \label{thm:weighted_matroid_intersection}
\end{theorem}

\begin{proof}
  Let $w^{W} = (w^{W}_1, w^{W}_2)$ with $w^{W}_i(x) = \frac{W}{2}$ for
  each $x \in V$
  and the initial common basis $S^{(0)}$ obtained via the unweighted
  matroid intersection algorithm be a $W$-solution of $w$.
  Repeatedly apply \Cref{thm:main} for $O(\log{rW})$ iterations to obtain a
  $\frac{1}{2r}$-solution $(w^{\frac{1}{2r}}, S^{*})$.
  For each $S \in \mathcal{I}_1 \cap \mathcal{I}_2$, we have
  \[ w(S) \leq w^{\frac{1}{2r}}(S) \leq w^{\frac{1}{2r}}(S^{*}) \leq w(S^{*}) + r \cdot \frac{1}{2r} < w(S^{*}) + 1. \]
  Since $w(S)$ and $w(S^{*})$ are integers, $S^{*}$ is a maximum-weight common
  basis.
  The algorithm runs in $O(nr^{3/4}\log{n}\log{(rW)} \cdot \trank)$ time.
  The theorem is proved.
\end{proof}

The rest of the paper proves \Cref{thm:main}.

\section{The Algorithm}

As in~\cite{fujishige1995} and~\cite{shigeno1995},
the algorithm of \Cref{thm:main} consists of the following two parts.

\subsection{Weight Adjustment}\label{sec:adjustment}

The first part of the algorithm is the following subroutine which computes two
bases $S_1$ and $S_2$ with a large enough intersection.
This part is essentially the same as Shigeno and Iwata's algorithm
\cite{shigeno1995},
except that we replace the fundamental (co-)circuit queries in it with calls to
\Cref{lemma:find_exchange}.

\begin{lemma}
  Given a $2\epsilon$-solution $(w^{2\epsilon}, S^\prime)$ and a parameter
  $1 \leq k \leq r$,
  it takes $O(nk\log{n} \cdot \trank)$ time to obtain an
  $\epsilon$-partial-solution $(w^{\epsilon}, S)$ with
  $|S_1 \cap S_2| \geq \left(1 - \frac{O(1)}{k}\right)r$.
  \label{lemma:weight_adjustment}
\end{lemma}

Since the algorithm and analysis are essentially the same as in
\cite{shigeno1995},
here we only describe how we can obtain $S_1$ and $S_2$ in the desired time
bound.
Please refer to~\cite{shigeno1995} or \Cref{lemma:large_intersection}
in \Cref{appendix:omitted_proofs} for the proof of
$|S_1 \cap S_2| \geq \left(1 - \frac{O(1)}{k}\right)r$.

\paragraph{Algorithm of~\Cref{lemma:weight_adjustment}.}
Let $w^{\epsilon} = (w^{2\epsilon}_1, w - w^{2\epsilon}_1 + \epsilon)$ be
the initial $\epsilon$-splitting and $S_i$ be the $w^{\epsilon}_i$-maximum
basis of $\mathcal{M}_i$ obtained by \Cref{lemma:greedy} in
$O(n\log{n} + n\trank)$ time for each $i \in \{1, 2\}$.
Let $p(x) = 0$ for each $x \in V$.
Repeat the following \emph{weight adjustment} for an arbitrary
$x \in S_1 \setminus S_2$ with $p(x) < k$ until such an $x$ becomes
non-existent.
\begin{itemize}
\item If $w^{\epsilon}(x) = w(x) + \epsilon$, then set
  $w^{\epsilon}_1(x) \gets w^{\epsilon}_1(x) - \epsilon$.
  Apply \Cref{lemma:find_exchange} to obtain a $y \in V \setminus S_1$
  maximizing $w^{\epsilon}_1(y)$ such that
  $(S_1 \setminus \{x\}) \cup \{y\} \in \mathcal{I}_1$.
  If $w^{\epsilon}_1(x) < w^{\epsilon}_1(y)$, then set
  $S_1 \gets (S_1 \setminus \{x\}) \cup \{y\}$.
\item Otherwise, set $p(x) \gets p(x) + 1$ and
  $w^{\epsilon}_2(x) \gets w^{\epsilon}_2(x) + \epsilon$.
  Apply \Cref{lemma:find_exchange} to obtain a $y \in S_2$ minimizing
  $w^{\epsilon}_2(y)$ such that
  $(S_2 \setminus \{y\}) \cup \{x\} \in \mathcal{I}_2$.
  If $w^{\epsilon}_2(x) > w^{\epsilon}_2(y)$, then set
  $S_2 \gets (S_2 \setminus \{y\}) \cup \{x\}$.
\end{itemize}
Since $p(x)$ is only incremented when $x \in S_1 \setminus S_2$, we have
$p(x) \leq k$ for each $x \in V$ when the procedure terminates.
Apparently, $w^{\epsilon}(x)$ oscillates between $w(x)$ and $w(x) + \epsilon$,
and thus the number of weight adjustments for $x$ is bounded by $2p(x)$.
We also have that $S_i$ remains $w^{\epsilon}_i$-maximum for each
$i \in \{1, 2\}$ due to the potential exchange of $x$ and $y$ after the
adjustment.
Each weight adjustment takes $O(\trank\log{n})$ time by
\Cref{lemma:find_exchange}, hence the total running time is
$O(nk\log{n} \cdot \trank)$.

\subsection{Augmentation}\label{sec:augmentation}

With $S_1$ and $S_2$ obtained from \Cref{lemma:weight_adjustment}, we then run
``few'' augmentations to make these two bases equal.
To do so, we need the following notion of exchange graphs, which is slightly
different compared to previous algorithms for unweighted matroid intersection
(e.g., \cite{blikstad2021,chakrabarty2019,cunningham1986,lawler1975}).

\paragraph{Exchange Graph.}
Let $(w^{\epsilon}, S)$ be an $\epsilon$-partial-solution of $w$ with
$S_1 \neq S_2$.
The \emph{exchange graph} with respect to $(w^{\epsilon}, S)$ is a
weighted directed multi-graph $G_{w^{\epsilon}, S} = (V \cup \{s, t\}, E)$
with $s, t \not\in V$ and $E = E_1 \cup E_2 \cup E_s \cup E_t$, where

\begin{align*}
  E_1 &= \{(x, y, w^{\epsilon}_1(x) - w^{\epsilon}_1(y))
    \mid x \in S_1, y \not\in S_1,\;\text{and}\;(S_1 \setminus \{x\}) \cup \{y\} \in \mathcal{I}_1\}, \\
  E_2 &= \{(y, x, w^{\epsilon}_2(x) - w^{\epsilon}_2(y))
    \mid x \in S_2, y \not\in S_2,\;\text{and}\;(S_2 \setminus \{x\}) \cup \{y\} \in \mathcal{I}_2\}, \\
  E_s &= \{(s, x, 0) \mid x \in S_1 \setminus S_2\},\;\text{and} \\
  E_t &= \{(x, t, 0) \mid x \in S_2 \setminus S_1\}.
\end{align*}

Since $S_i$ is $w^{\epsilon}_i$-maximum for each $i \in \{1, 2\}$, all edge
weights are non-negative.
Note that this definition of exchange graph is a simplified
version of the \emph{auxiliary graph} defined by Fujishige and Zhang
\cite{fujishige1995}
to solve the more generalized \emph{independent assignment}
problem\footnote{Specifically, given a bipartite graph $G = (V_1 \cup V_2, E)$
with $V_1$ and $V_2$ being copies of $V$ and two matroids
$\mathcal{M}_1 = (V, \mathcal{I}_1)$,
$\mathcal{M}_2 = (V, \mathcal{I}_2)$ on $V$, the independent assignment problem
aims to find the largest $S_1 \in \mathcal{I}_1$ and $S_2 \in \mathcal{I}_2$
such that $G$ admits a perfect matching between $S_1 \subseteq V_1$
and $S_2 \subseteq V_2$.
Analogously, the weighted version of the problem wants to find $S_1$ and $S_2$
such that the weight of the
maximum-weight perfect matching between $S_1$ and $S_2$ is maximized.
Clearly, the (weighted) matroid intersection problem is a special case of the
(weighted) independence assignment problem with $E = \{(v, v) \mid v \in V\}$.}.
We have the following properties of the exchange graph, for which we also
provide simplified and more direct proofs for self-containedness in
\Cref{appendix:omitted_proofs}.

\begin{lemma}[\cite{fujishige1995}; See \Cref{appendix:omitted_proofs}]
  $G_{w^{\epsilon}, S}$ admits an $st$-path.
  \label{lemma:has_path}
\end{lemma}

Let $d(x)$ be the $sx$-distance in $G_{w^{\epsilon}, S}$ for each $x \in V$
(set $d(x)$ to a large number if $x$ is unreachable from $s$; see
\Cref{sec:bounding} for the exact value) and
$P$ be the shortest $st$-path with the least number of edges.

\begin{lemma}[\cite{fujishige1995}; See \Cref{appendix:omitted_proofs}]
  $\hat{S}_1 = (S_1 \setminus \tail(P \cap E_1)) \cup \head(P \cap E_1)$
  and
  $\hat{S}_2 = (S_2 \setminus \head(P \cap E_2)) \cup \tail(P \cap E_2)$
  are a $\hat{w}^{\epsilon}_1$-maximum and $\hat{w}^{\epsilon}_2$-maximum
  basis of $\mathcal{M}_1$ and $\mathcal{M}_2$, respectively, where
  $\hat{w}^{\epsilon}_1(x) = w^{\epsilon}_1(x) + d(x)$ and
  $\hat{w}^{\epsilon}_2(x) = w^{\epsilon}_2(x) - d(x)$ for each $x \in V$.
  In particular, $(\hat{w}^{\epsilon}, \hat{S})$ with
  $\hat{w}^{\epsilon} = (\hat{w}^{\epsilon}_1, \hat{w}^{\epsilon}_2)$
  and $\hat{S} = (\hat{S}_1, \hat{S}_2)$ is an $\epsilon$-partial-solution.
  Moreover, we have $|\hat{S}_1 \cap \hat{S}_2| > |S_1 \cap S_2|$.
  \label{lemma:augment}
\end{lemma}

With the above properties and \Cref{lemma:weight_adjustment}, we finish our
algorithm with the following shortest-path procedure.
Note that in order to make the algorithm subquadratic, we do not construct
the exchange graphs explicitly.
Nevertheless, we show that a partial construction suffices to compute
the shortest-path trees in them.

\begin{lemma}
  It takes $O(n\sqrt{r}\log{n} \cdot \trank)$ time to obtain
  $d(x)$ for each $x \in V$ and the shortest $st$-path with the least number of
  edges in $G_{w^{\epsilon}, S}$.
  \label{lemma:dijkstra}
\end{lemma}

We are now ready to prove~\Cref{thm:main}.

\begin{proof}[Proof of \Cref{thm:main}]
  Apply \Cref{lemma:weight_adjustment} with $k = r^{3/4}$ to obtain an
  $\epsilon$-partial-solution $(w^{\epsilon}, S)$ of $w$ such that
  $|S_1 \cap S_2| \geq r - O(r^{1/4})$ in
  $O(nr^{3/4}\log{n} \cdot \trank)$ time.
  For $O(r^{1/4})$ iterations, apply \Cref{lemma:dijkstra,lemma:augment} to
  obtain $(\hat{w}^{\epsilon}, \hat{S})$ with $\hat{S}_1$ and $\hat{S}_2$
  having a larger intersection than $S_1$ and $S_2$ do,
  and set $(w^{\epsilon}, S) \gets (\hat{w}^{\epsilon}, \hat{S})$ until
  $S_1 = S_2$.
  This takes overall $O(nr^{3/4}\log{n} \cdot \trank)$ time
  as well.
  Note that $w^{\epsilon}(x) = w^{\epsilon}_1(x) + w^{\epsilon}_2(x)$
  remains the same, completing the proof.
\end{proof}

The remainder of this section proves \Cref{lemma:dijkstra}.
For the ease of notation, we abbreviate $G_{w^{\epsilon}, S}$ as $G$.

Intuitively, we would like to run Dijkstra's algorithm on $G$ to build a
shortest-path tree.
However, na\"{i}ve implementation takes $O(nr)$ time since we might need to
relax $O(nr)$ edges.
This is unlike the BFS algorithm of Chakrabarty et al.~\cite{chakrabarty2019}
for the unweighted case, where we can immediately mark all out-neighbors of the
current vertex as ``visited'', leading to a near-linear running time.
To speed things up, note that using \Cref{lemma:find_exchange}, for a vertex
$x$, we can efficiently find the vertex which is ``closest'' to $x$.
Let $F$ denote the set of visited vertices whose exact distances are known.
The closest unvisited vertex to $F$ must be closest to some $x \in F$.
Therefore, in each iteration, it suffices to only relax the ``shortest'' edge
from each $x \in F$.
This can be done efficiently by maintaining a set of ``recently visited''
vertices $B$ of size roughly $\sqrt{n}$ and computing the distance estimate from
$F \setminus B$ to all unvisited vertices\footnote{In the actual algorithm,
we maintain two buffers instead of one to further improve the running time
to $\tilde{O}(n\sqrt{r})$ from $\tilde{O}(n\sqrt{n})$.
This makes our weighted matroid intersection algorithm $o(nr)$ as opposed to
just $o(n^2)$.}.
In each iteration, we relax the shortest edge from each $x \in B$, and now the
vertex with the smallest distance estimate is closest to $F$ and therefore we
include it into $B$ (and thus $F$).
When $B$ grows too large, we clear $B$ and recompute the distance estimates from
$F$ in $\tilde{O}(n)$ queries.
This leads to a subquadratic algorithm.
We now prove the lemma formally.

\begin{proof}[Proof of \Cref{lemma:dijkstra}]
  The algorithm builds a shortest-path tree of $G$ using Dijkstra's algorithm.
  We maintain a distance estimate $\hat{d}(x)$ for each $x \in V \cup \{s, t\}$.
  Initially, $\hat{d}(x) = 0$ for each $x \in (S_1 \setminus S_2) \cup \{s\}$
  and $\hat{d}(x) = \infty$ for other vertices.
  Edge set $E_t$ is only for the convenience of defining an $st$-path and thus
  we may ignore it here.
  Let $F$ be the set of \emph{visited} vertices whose distance estimates are
  correct, i.e., $d(x) = \hat{d}(x)$ holds for each $x \in F$.
  Initially, $F = \{s\}$.
  The algorithm runs in at most $n$ iterations, and in the $t$-th iteration,
  we visit a new vertex $v_{t}$ such that $d(v_{t}) = \hat{d}(v_t)$ and
  $d(v_{t}) \leq d(v)$ for each $v \not\in F$.
  We maintain two \emph{buffers} $B_1 \subseteq F \cap S_1$ and
  $B_2 \subseteq F \cap \bar{S_2}$ containing vertices in $S_1$ and
  $\bar{S_2}$ that are visited ``recently''.
  That is, after the $t$-th iteration, we have
  $B_1 = \{v_i, \ldots, v_{t}\} \cap S_1$ or $B_1 = \emptyset$ and
  $B_2 = \{v_j, \ldots, v_{t}\} \cap \bar{S_2}$ or $B_2 = \emptyset$ for some
  $i, j \leq t$.
  Recall that $E_1$ and $E_2$ are the edges in $G$ that correspond to
  exchange relations in $\mathcal{M}_1$ and $\mathcal{M}_2$, respectively.
  For each $i \in \{1, 2\}$ and edge $e = (x, y)$, let
  $w_i(x, y) = |w^{\epsilon}_i(x) - w^{\epsilon}_i(y)|$ be the edge weight
  of $e$ in $E_i$ if $e \in E_i$ and $w_i(x, y) = \infty$ otherwise.
  Let $E({B_i}) = \{(x, y) \in E_i \mid x \in B_i\;\text{and}\;y \not\in F\}$
  be edges in $E_i$ directing from $B_i$ to $V \setminus F$
  and $E(B) = E({B_1}) \cup E({B_2})$.
  Let
  $E({F}) = \{(x, y) \in E_1 \cup E_2 \mid x \in F\;\text{and}\;y \not\in F\}$.
  For $v \in V \setminus F$ and edge set $E^\prime$ such that
  $\tail(E^\prime) \subseteq F$, let
  $\tilde{d}(v, E^\prime) = \min_{(x, v) \in E^\prime}\{d(x) + w(x, v)\}$ be the
  shortest distance to $v$ ``relaxed'' by edges in $E^\prime$
  (recall that $w(x, v)$ is the weight of the edge $(x, v)$).
  We maintain the following invariants after each iteration of the algorithm
  except the last one.
  \begin{enumerate}[label=(\roman*)]
    \item\label{I1}
      $d(v) \leq \hat{d}(v) \leq \tilde{d}(v, E(F) \setminus E(B))$ holds
      for each $v \in V \setminus F$.
    \item\label{I2} There exists a $v \in V \setminus F$ such that
      $\hat{d}(v) = d^{*}_F := \min_{u \in V \setminus F}\{\tilde{d}(u, E(F))\}$.
  \end{enumerate}
  Intuitively, Invariant~\ref{I1} asserts that all edges in
  $E(F) \setminus E(B)$ are ``relaxed'' while Invariant~\ref{I2} ensures that
  the distance estimate of the target vertex, i.e., one with the shortest
  distance from $s$, is correct.
  Initially, both invariants are satisfied since
  $\hat{d}(x) = 0$ holds for each $x \in S_1 \setminus S_2$.
  We maintain a priority queue $Q$ containing vertices in $V \setminus F$
  ordered by $\hat{d}(\cdot)$.
  In the $t$-th iteration, let $v_{t}$ be the vertex $v$ with the smallest
  $\hat{d}(v)$.
  By Invariants~\ref{I1} and~\ref{I2}, we have $\hat{d}(v_{t}) = d^{*}_F$
  and thus $d(v_{t}) \leq d(v^\prime)$ holds for each
  $v^\prime \in V \setminus F$ according to Dijkstra's algorithm.
  As such, we push $v_{t}$ into $F$ and update $B_1$, $B_2$ appropriately by
  checking if $v_{t}$ belongs to $S_1$ and $\bar{S_2}$.
  Now, we would like to modify $\hat{d}(v)$ for some $v \in V \setminus F$ so
  that both invariants remain true.
  For each $i \in \{1, 2\}$, depending on the size of $B_i$, we perform one of
  the following.
  
  \begin{enumerate}
    \item\label{case1} If $|B_i| \geq \sqrt{r}$, then we compute
      $\tilde{d}_i(v) = \tilde{d}(v, E(B_1))$ and set
      $\hat{d}(v) \gets \min(\hat{d}(v), \tilde{d}_i(v))$ for
      each $v \in V \setminus F$ using \Cref{lemma:relax} below.
      For $i = 1$, by definition of $G$, $\head(E_1) \subseteq V \setminus S_1$
      and thus we only need to compute $\tilde{d}_i(v)$ for
      $v \in V \setminus S_1$, and therefore \Cref{lemma:relax} takes
      $O(n\log{n} \cdot \trank)$ time.
      For $i = 2$, similarly, $\head(E_2) \subseteq S_2$ and thus
      we only need to compute $\tilde{d}_i(v)$ for $v \in S_2$,
      taking $O(r\log{n} \cdot \trank)$ time.
      Then, we set $B_i \gets \emptyset$, and the above modification ensures
      that Invariant~\ref{I1} holds since
      $\tilde{d}(v, E(F)) = \min(\tilde{d}(v, E(F) \setminus E(B)), \tilde{d}(v, E(B)))$.
    \item\label{case2} If $|B_i| < \sqrt{r}$, then we do not clear $B_i$ and
      therefore Invariant~\ref{I1} trivially holds.
      For each $b \in B_i$, we find a $v_b \in V \setminus F$ minimizing
      $d(b) + w_i(b, v_b)$ via \Cref{lemma:find_exchange} as follows.
      If $i = 1$, then we have
      $w_1(b, v_b) = w^{\epsilon}_1(b) - w^{\epsilon}_1(v_b)$,
      and thus we find the $v_b$ maximizing $w^{\epsilon}_1(v_b)$ such that
      $(S_1 \setminus \{b\}) \cup \{v_b\} \in \mathcal{I}_1$.
      If $i = 2$, then
      $w_2(b, v_b) = w^{\epsilon}_2(v_b) - w^{\epsilon}_2(b)$,
      and thus we find the $v_b$ minimizing $w^{\epsilon}_2(v_b)$ such that
      $(S_2 \setminus \{v_b\}) \cup \{b\} \in \mathcal{I}_2$.
      Then, we set $\hat{d}(v_b) \gets \min(\hat{d}(v_b), d(b) + w_i(b, v_b))$
      and update $v_b$'s position in $Q$ appropriately.
      This takes $O(\sqrt{r}\log{n}\cdot\trank)$ time.
  \end{enumerate}
  
  In both cases, as argued above, Invariant~\ref{I1} holds.
  We argue that Invariant~\ref{I2} holds after the iteration as well.
  Let $B_1^{(t)}$ be $B_1$ after the $t$-th iteration and define
  $B_2^{(t)}$ and $F^{(t)}$ similarly.
  Let $E(B^{(t)})$ denote $E(B_1^{(t)}) \cup E(B_2^{(t)})$.
  Let $v^{*} = \argmin_{v \in V \setminus F^{(t)}}\{\tilde{d}(v, E(F^{(t)}))\}$
  be an unvisited vertex after the $t$-th iteration with the smallest distance
  from $s$ and let $e^{*} = (u, v^{*})$ be the edge such that $u \in F^{(t)}$
  and $d(v^{*}) = d(u) + w(e^{*})$.
  That is, $e^{*}$ is the edge connecting $v^{*}$ and its parent in the
  shortest-path tree.
  If $e^{*} \in E(F^{(t-1)}) \setminus E(B^{(t-1)})$, then Invariant~\ref{I2}
  trivially follows from the end of the $(t - 1)$-th iteration.
  Otherwise, we must have either $e^{*} \in E(B_1^{(t-1)})$ or
  $e^{*} \in E(B_2^{(t-1)})$.
  Without loss of generality, let's assume $e^{*} \in E(B_1^{(t-1)})$.
  If $|B_1^{(t-1)}| + 1 \geq \sqrt{r}$ (i.e., Case~\ref{case1}), then after
  setting $B_1^{(t)} \gets \emptyset$, Invariant~\ref{I2} follows from the
  fact the Invariant~\ref{I1} holds for $v^{*}$ and
  $\tilde{d}(v^{*}, E(F^{(t)}) \setminus E(B^{(t)})) \leq d(u) + w(e^{*})$ since
  $e^{*} \in E(F^{(t)}) \setminus E(B^{(t)})$.
  If $|B_1^{(t-1)}| + 1 < \sqrt{r}$ (i.e., Case~\ref{case2}), then there must
  exists a $b \in B_1^{(t)}$ such that
  $d(b) + w_1(b, v^{*}) = \min_{v}\{d(b) + w_1(b, v)\}$
  and thus we have at least one $v_b \in V \setminus F^{(t)}$ such that
  $\hat{d}(v_b) \leq d(b) + w_1(b, v_b) = d(v^{*})$.
  This shows that Invariant~\ref{I2} indeed holds after the $t$-th iteration.
  The correctness of the algorithm follows from the two invariants and
  the analysis of Dijkstra's algorithm.
  
  To bound the total running time, observe that for $B_1$, Case~\ref{case1}
  happens at most $O(r/\sqrt{r}) = O(\sqrt{r})$ times since $|S_1| = r$.
  Thus, it takes $O(n\sqrt{r}\log{n} \cdot \trank)$ time in
  total.
  Similarly, for $B_2$, Case~\ref{case1} happens at most $O(n/\sqrt{r})$ time,
  taking
  $O(n/\sqrt{r} \cdot r\log{n} \cdot \trank)
    = O(n\sqrt{r}\log{n} \cdot \trank)$
  time in total as well.
  For Case~\ref{case2}, each iteration takes
  $O(\sqrt{r}\log{n}\cdot\trank)$ time, contributing a total of
  $O(n\sqrt{r}\log{n} \cdot \trank)$ time.
  As a result, the algorithm runs in
  $O(n\sqrt{r}\log{n} \cdot \trank)$ time, as claimed.

  Finally, it is easy to maintain balanced binary search trees of elements
  ordered by $u_1$, $u_2$, $\hat{d} + u_1$, and $\hat{d} - u_2$ in
  $O(n\log{n})$ time throughout the procedure so that
  \Cref{lemma:find_exchange} can be applied without overhead.
  The shortest $st$-path can also be easily recovered by maintaining the
  optimal parent in the shortest-path tree for each vertex.
  This proves the lemma.
\end{proof}

\begin{lemma}
  For each $i \in \{1, 2\}$,
  given $B_i \subseteq F$ and $R \subseteq V \setminus F$,
  it takes
  takes $O(|R| \log{n} \cdot \trank)$ time to compute
  $\tilde{d}(v, E(B_i))$ for all $v \in R$.
  \label{lemma:relax}
\end{lemma}

\begin{proof}
  For $i = 1$ and $e = (b, v) \in E(B_i)$, we have
  $w(e) = w^{\epsilon}_1(b) - w^{\epsilon}_1(v)$.
  Therefore, $d(v, E(B_1))$ can be computed by finding the $b \in B$ with the
  smallest $d(b) + w^{\epsilon}_1(b)$ such that
  $(S_1 \setminus \{b\}) \cup \{v\} \in \mathcal{I}_1$ via
  \Cref{lemma:find_exchange}.
  Similarly, for $i = 2$, we have
  $w(e) = w^{\epsilon}_2(v) - w^{\epsilon}_2(b)$,
  and thus $d(v, E(B_2))$ can be computed by finding the $b \in B$ with the
  smallest $d(b) - w^{\epsilon}_2(b)$.
  The lemma simply follows by calling \Cref{lemma:find_exchange} once for each
  $v \in R$.
\end{proof}

\subsection{Bounding the Numbers}\label{sec:bounding}

Finally, to conclude the analysis of our algorithm, we argue that the numbers
such as $w^{\epsilon}_1(x)$ and $w^{\epsilon}_2(x)$ are bounded by
$\tilde{O}(\poly(nW))$ so that the number of bits needed to store them and the
time for a single arithmetic operation only grow by a constant factor.
In the weight adjustment stage, each number is adjusted at most
$O(r)$ times and each adjustment changes the number by at most $O(W)$
since $\epsilon$ is at most $W$.
Therefore, the accumulative change to a number via weight adjustments is
at most $O(\poly(nW))$.
For growth incurred by augmentations, we first assume that all vertices are
reachable from $s$ in $G_{w^{\epsilon}, S}$.
Consider a single run of \Cref{lemma:dijkstra} and fix an $x \in V$.
Let $P_x = \{s, v_1, \ldots, v_k\}$ with $v_k = x$ be the shortest
$sx$-path in $G_{w^{\epsilon}, S}$.
Suppose that $(v_1, v_2), (v_{k-1}, v_k) \in E_1$, then by definition,
we have
\begin{align*}
  d(x)
    &= w^{\epsilon}_1(v_1) - w^{\epsilon}_1(v_2) + w^{\epsilon}_2(v_3) - w^{\epsilon}_2(v_2) + \cdots + w^{\epsilon}_1(v_{k-1}) - w^{\epsilon}_1(v_k) \\
    &\leq w^{\epsilon}_1(v_1) - w(v_2) + (w(v_3) + \epsilon) + \cdots + (w(v_{k-1}) + \epsilon) - w^{\epsilon}_1(v_k) \\
    &\leq w^{\epsilon}_1(v_1) - w^{\epsilon}_1(v_k) + \left(\sum_{i = 2}^{k-1}(-1)^{i+1}w(v_i)\right) + nW.
\end{align*}
Since $\hat{w}^{\epsilon}_1(x) = w^{\epsilon}_1(x) + d(x)$ and
$\hat{w}^{\epsilon}_2(x) = w^{\epsilon}_2(x) - d(x)$ as defined in
\Cref{lemma:augment},
we have
\begin{eqnarray}
  |\hat{w^{\epsilon}}_1(x)| \leq |w^{\epsilon}_1(v_1)| + 2nW & \text{and} & |\hat{w}^{\epsilon}_2(x)| \leq |w^{\epsilon}_1(v_1)| + 2nW.
  \label{eq:abs_bound_1}
\end{eqnarray}
Similarly, if $(v_{k-1}, v_k) \in E_2$, then
\begin{align*}
  d(x)
    &= w^{\epsilon}_1(v_1) - w^{\epsilon}_1(v_2) + w^{\epsilon}_2(v_3) - w^{\epsilon}_2(v_2) + \cdots + w^{\epsilon}_2(v_k) - w^{\epsilon}_2(v_{k-1}) \\
    &\leq w^{\epsilon}_1(v_1) - w(v_2) + (w(v_3) + \epsilon) + \cdots + (w(v_{k-2}) + \epsilon) - w(v_{k-1}) + w^{\epsilon}_2(v_k) \\
    &\leq w^{\epsilon}_1(v_1) + w^{\epsilon}_2(v_k) + \left(\sum_{i=2}^{k-1}(-1)^{i+1}w(v_i)\right) + nW,
\end{align*}
implying \eqref{eq:abs_bound_1} as well.
The case when $(v_1, v_2) \in E_2$ holds similarly, except now we have
\begin{eqnarray}
|\hat{w}^{\epsilon}_1(x)| \leq |w^{\epsilon}_2(v_1)| + 2nW & \text{and} & |\hat{w}^{\epsilon}_2(x)| \leq |w^{\epsilon}_2(v_1)| + 2nW.
\label{eq:abs_bound_2}
\end{eqnarray}

Since the number of augmentations is $\tilde{O}(r^{1/4})$, we indeed have
that
$|w^{\epsilon}_1(x)| = |w^{\epsilon}_2(x)| = O(\poly(nW)) = \Theta((nW)^{k})$
for some constant $k$.
For the case where some vertex $x$ is not reachable from $s$, we can
simply set $d(x)$ to some $c(nW)^{k+1}$ for a large enough constant
$c$ and the desired bound still holds.


\section{Concluding Remarks}\label{sec:conclusion}
We present a simple subquadratic algorithm for weighted matroid intersection
under the rank oracle model, providing a partial yet affirmative answer to one
of the open problems raised by Blikstad et al.~\cite{blikstad2021}.
Whether the same is achievable under the independence oracle model remains open.
It seems that our techniques for computing shortest-path trees do not solely
result in a subquadratic augmenting-path algorithm under the independence
oracle.
Removing the dependence on $\log{W}$ and making the algorithm run in
strongly-polynomial time is also of interest.
Finally, as noted in~\cite{blikstad2021}, there were very few non-trivial lower
bound results for matroid intersection.
It would be helpful to see if there is any super-linear lower bound on the
number of queries for these problems or even for computing shortest-path trees
in the exchange graphs under either oracle model.

\section*{Acknowledgements}
I would like to thank Prof. Hsueh-I Lu for advising the
project and helpful suggestions on the writing and notation, Brian Tsai for
proof-reading an initial draft of this paper, and the anonymous reviewers of
ISAAC 2022 for their useful comments.

\bibliography{main}

\appendix

\section{Omitted Proofs}
\label{appendix:omitted_proofs}

\subsection{Proofs of Lemmas in \Cref{sec:adjustment}}

For self-containedness, we include the proof that the two bases obtained in the
weight adjustment phase have a large intersection by Shigeno and Iwata
\cite{shigeno1995} here.

\begin{lemma}[\cite{shigeno1995}]
  Let $S_1$ and $S_2$ be obtained from the procedure described in
  \Cref{lemma:weight_adjustment}.
  Then, $|S_1 \cap S_2| \geq \left(1 - \frac{O(1)}{k}\right)r$.
  \label{lemma:large_intersection}
\end{lemma}

\begin{proof}
  Let $p(S)$ denote $\sum_{x \in S}p(x)$.
  Observe that an element is never moved to $S_2 \setminus S_1$ during
  weight adjustments, and therefore we have $p(S_2 \setminus S_1) = 0$
  and $p(S_1 \setminus S_2) = p(S_1) - p(S_2)$.
  Recall that $S^\prime$ is a common basis such that
  $(w^{2\epsilon}, S^\prime)$ is a $2\epsilon$-solution.
  Since $p(x)$ equals the number of adjustments of $w^{\epsilon}_2(x)$ and
  each such adjustment is preceded by an adjustment of $w^{\epsilon}_1(x)$,
  we have
  \[ p(x) \cdot \epsilon = w^{\epsilon}_2(x) - (w(x) - w^{2\epsilon}_1(x)) \leq w^{2\epsilon}_1(x) - w^{\epsilon}_1(x) \]
  for each $x \in V$.
  Thus,
  \begin{align*}
    p(S_1 \setminus S_2) \cdot \epsilon
      &= \left(p(S_1) - p(S_2)\right) \cdot \epsilon \\
      &\leq \left(w^{2\epsilon}_1(S_1) - w^{\epsilon}_1(S_1)\right)
        - \left(w^{\epsilon}_2(S_2) - w(S_2) + w^{2\epsilon}_1(S_2)\right) \\
      &\stackrel{\text{(a)}}{\leq} w^{2\epsilon}_1(S_1) - w^{\epsilon}_1(S^\prime) - w^{\epsilon}_2(S^\prime) + w(S_2) - w^{2\epsilon}_1(S_2) \\
      &\stackrel{\text{(b)}}{\leq} w^{2\epsilon}_1(S_1) - w(S^\prime) + w(S_2) - w^{2\epsilon}_1(S_2),
  \end{align*}
  where (a) is because $S_i$ is $w^{\epsilon}_i$-maximum for each
  $i \in \{1, 2\}$ and
  (b) is because $w(S^\prime) \leq w^{\epsilon}(S^\prime)$ as $w^{\epsilon}$
  is an $\epsilon$-splitting.
  Since $(w^{2\epsilon}, S^\prime)$ is a $2\epsilon$-solution,
  $w^{2\epsilon}_2(S) - 2\epsilon r \leq w(S) - w^{2\epsilon}_1(S) \leq w^{2\epsilon}_2(S)$
  holds for each basis $S$.
  This combined with the fact that $S^\prime$ is $w^{2\epsilon}_i$-maximum
  for each $i \in \{1, 2\}$ implies
  \[
    p(S_1 \setminus S_2) \cdot \epsilon
      \leq 2\epsilon r - w^{2\epsilon}_2(S^\prime) + w^{2\epsilon}_2(S_2)
      \leq 2\epsilon r
      \implies p(S_1 \setminus S_2) \leq 2r.
  \]
  When the algorithm terminates, we have $p(x) = k$ for all
  $x \in S_1 \setminus S_2$, implying
  \[
    p(S_1 \setminus S_2) = |S_1 \setminus S_2| \cdot k
      \leq 2r \implies |S_1 \setminus S_2| \leq \frac{2r}{k}.
  \]
  As a result,
  \[
    |S_1 \cap S_2| = r - |S_1 \setminus S_2|
      \geq \left(1 - \frac{O(1)}{k}\right) r.
  \]
\end{proof}

\subsection{Proofs of Lemmas in \Cref{sec:augmentation}}

In this section, we prove the properties of the exchange graphs.
The proofs for the more generalized auxiliary graph given by Fujishige and
Zhang can be found in~\cite{fujishige1995}. 

To prove \Cref{lemma:has_path}, it would be more convenient to refer to the
following definition of a directed bipartite graph based on exchange
relationships, which is heavily used in unweighted matroid intersection
algorithms.
For $S \in \mathcal{I}_1 \cap \mathcal{I}_2$, let
$\tilde{G}_S = (V \cup \{s, t\}, \tilde{E})$ with $s, t \not\in V$ denote
the directed graph with
$\tilde{E} = \tilde{E}_1 \cup \tilde{E}_2 \cup \tilde{E}_s \cup \tilde{E}_t$,
where
\begin{align*}
  \tilde{E}_1 &= \{(x, y) \mid x \in S, y \not\in S,\;\text{and}\;(S \setminus \{x\}) \cup \{y\} \in \mathcal{I}_1\}, \\
  \tilde{E}_2 &= \{(y, x) \mid x \in S, y \not\in S,\;\text{and}\;(S \setminus \{x\}) \cup \{y\} \in \mathcal{I}_2\}, \\
  \tilde{E}_s &= \{(s, x) \mid S \cup \{x\} \in \mathcal{I}_1 \},\;\text{and} \\
  \tilde{E}_t &= \{(x, t) \mid S \cup \{x\} \in \mathcal{I}_2 \}.
\end{align*}

\begin{lemma}[\cite{lawler1975}]
  $\tilde{G}_S$ for $|S| < r$ admits an $st$-path.
  \label{lemma:has_path_unweighted}
\end{lemma}

We will use the existence of an $st$-path in $\tilde{G}_{\tilde{S}}$ to prove
that such a path exists in $G_{w^{\epsilon}, S}$, for $\tilde{S} = S_1 \cap S_2$.
The following claims certify that $\tilde{G}_{\tilde{S}}$ and $G_{w^{\epsilon}, S}$
are almost the same.

\begin{claim}
  Let $\mathcal{M} = (V, \mathcal{I})$ be a matroid,
  $S \subseteq S^\prime \in \mathcal{I}$, $x \in S$, and $y \not\in S^\prime$
  such that $(S \setminus \{x\}) \cup \{y\} \in \mathcal{I}$ but
  $S \cup \{y\} \not\in \mathcal{I}$, then
  $(S^\prime \setminus \{x\}) \cup \{y\} \in \mathcal{I}$.
  \label{claim:still_edge}
\end{claim}

\begin{proof}
  Let $C$ be the unique circuit in $S \cup \{y\}$.
  Since $C \subseteq S^\prime \cup \{y\}$ and $S^\prime \cup \{y\}$ has only one
  circuit, $C$ is the unique circuit in $S^\prime \cup \{y\}$ as well.
  Moreover, $(S \setminus \{x\}) \cup \{y\} \in \mathcal{I}$ if and only if
  $x \in C$ and therefore
  $(S^\prime \setminus \{x\}) \cup \{y\} \in \mathcal{I}$.
\end{proof}

\begin{claim}
  Let $\mathcal{M} = (V, \mathcal{I})$ be a matroid,
  $S \subseteq S^\prime \in \mathcal{I}$ where $S^\prime$ is a basis of
  $\mathcal{M}$, and $x \not\in S^\prime$ such that
  $S \cup \{x\} \in \mathcal{I}$.
  Then, there exists a $y \in S^\prime \setminus S$ such that
  $(S^\prime \setminus \{y\}) \cup \{x\} \in \mathcal{I}$.
  \label{claim:extend}
\end{claim}

\begin{proof}
  Let $S^{\prime\prime}$ be an arbitrary basis of $\mathcal{M}$ that contains
  $S \cup \{x\}$.
  Since $x \in S^{\prime\prime} \setminus S^\prime$, by the strong exchange
  property (see, e.g.,~\cite[Theorem 39.6]{schrijver2003}) of bases, there
  exists a
  $y \in S^\prime \setminus S^{\prime\prime} \subseteq S^\prime \setminus S$
  such that $(S^\prime \setminus \{y\}) \cup \{x\} \in \mathcal{I}$, completing
  the proof.
\end{proof}

We are now ready to prove \Cref{lemma:has_path}.

\begin{proof}[Proof of \Cref{lemma:has_path}]
  Let $\tilde{P} = \{s, v_1, \ldots, v_k, t\}$ be the shortest $st$-path in
  $\tilde{G}_{\tilde{S}}$ for $\tilde{S} = S_1 \cap S_2$.
  The existence of such a path is guaranteed by
  \Cref{lemma:has_path_unweighted}.
  We have $\tilde{S} \cup \{v_i\} \not\in \mathcal{I}_1$ and
  $\tilde{S} \cup \{v_i\} \not\in \mathcal{I}_2$ for each $1 < i < k$ since
  $\tilde{P}$ is the shortest path.
  For an odd $1 \leq i < k$, we have $v_i \not\in S$ and $v_{i+1} \in S$.
  If $v_i \not\in S_2 \setminus S_1$, then by \Cref{claim:still_edge}, we have
  $(v_{i}, v_{i+1}) \in E(G_{w^{\epsilon}, S})$.
  Similarly, for an even $1 \leq i < k$, if $v_{i+1} \not\in S_1 \setminus S_2$,
  then we have $(v_{i}, v_{i+1}) \in E(G_{w^{\epsilon}, S})$.
  Suppose that $v_1 \not\in S_1 \setminus S_2$, then by \Cref{claim:extend}, we
  can find a $v_0 \in S_1 \setminus S_2$ such that
  $(S_1 \setminus \{v_0\}) \cup \{v_1\} \in \mathcal{I}_1$.
  Similarly, if $v_k \not\in S_2 \setminus S_1$, then we can find a
  $v_{k+1} \in S_2 \setminus S_1$ such that
  $(S_2 \setminus \{v_{k+1}\}) \cup \{v_k\} \in \mathcal{I}_2$.
  Therefore, without loss of generality, we may assume that there exists the
  last vertex $v_i \in S_1 \setminus S_2$ and the first vertex
  $v_j \in S_2 \setminus S_1$ after $v_i$.
  Now, for each $i < k < j$, we have
  $v_k \not\in (S_1 \setminus S_2) \cup (S_2 \setminus S_1)$.
  Therefore, $(v_{k}, v_{k+1}) \in E(G_{w^{\epsilon}, S})$ holds for each
  $i \leq k < j$, and we obtain an $st$-path in $G_{w^{\epsilon}, S}$ as
  $P = \{s, v_i, \ldots, v_j, t\}$.
  This concludes the proof.
\end{proof}

Finally, to prove \Cref{lemma:augment}, we need the following results.

\begin{lemma}[{\cite[Proposition 2.4.1]{price2015}}]
  Given a matroid $\mathcal{M} = (V, \mathcal{I})$ and an $S \in \mathcal{I}$.
  Suppose that $(a_1, \ldots, a_p) \subseteq V \setminus S$ and
  $(b_1, \ldots, b_p) \subseteq S$ are two sequences satisfying the following
  conditions:
  \begin{enumerate}
    \item $(S \setminus \{b_i\}) \cup \{a_i\} \in \mathcal{I}$ for each
      $1 \leq i \leq p$ and
    \item $(S \setminus \{b_j\}) \cup \{a_i\} \not\in \mathcal{I}$ for each
      $1 \leq j < i \leq p$.
  \end{enumerate}
  Then,
  $(S \setminus \{b_1, \ldots, b_p\}) \cup \{a_1, \ldots, a_p\} \in \mathcal{I}$
  holds.
  \label{lemma:can_augment}
\end{lemma}

\begin{lemma}[{\cite[Lemma 2.4.2]{price2015}}]
  Let $\mathcal{M}$, $S$, $(a_1, \ldots, a_p)$, and $(b_1, \ldots, b_p)$ be the
  same as in \Cref{lemma:can_augment}.
  Let $S^\prime = (S \setminus \{b_1, \ldots, b_p\}) \cup \{a_1, \ldots, a_p\}$.
  For $x \in S^\prime$ and $y \in V \setminus S^\prime$, if
  $(S^\prime \setminus \{x\}) \cup \{y\} \in \mathcal{I}$
  but either $y \in S$ or $(S \setminus \{x\}) \cup \{y\} \not\in \mathcal{I}$,
  then there exists
  $1 \leq \ell \leq k \leq p$ such that
  $(S \setminus \{x\}) \cup \{a_k\} \in \mathcal{I}$ and
  either $b_\ell = y$ or $(S \setminus \{b_\ell\}) \cup \{y\} \in \mathcal{I}$.
  \label{lemma:new_edges}
\end{lemma}

In essence, \Cref{lemma:can_augment} captures the validity of an augmentation
while \Cref{lemma:new_edges} models the condition in which new exchange
relationships emerge in the augmented independent set.
The following three claims imply \Cref{lemma:augment}.
Recall that $P = \{s, v_1, \ldots, v_k, t\}$ is the shortest $st$-path with the
least number of edges and $d(x)$ is the $sx$-distance in $G_{w^{\epsilon}, S}$.

\begin{claim}
  $|\hat{S}_1 \cap \hat{S}_2| > |S_1 \cap S_2|$ holds.
\end{claim}

\begin{proof}
  If $k = 2$, then either $v_1 \in \hat{S}_1 \cap \hat{S}_2$ or
  $v_k \in \hat{S}_1 \cap \hat{S}_2$ must hold, depending on whether
  $(v_1,v_2) \in E_1$ or $(v_1,v_2) \in E_2$, and the claim trivially holds
  in this case.
  Thus, in the following, we assume that $k > 2$.
  Since $P$ is the shortest, we may assume that
  $v_i \in (S_1 \cap S_2) \cup (\bar{S_1} \cap \bar{S_2})$ holds for each
  $1 < i < k$.
  Also, for $1 < i < k - 1$, if $v_{i} \in S_1 \cap S_2$, then $v_{i+1}$ must be
  in $\bar{S_1} \cap \bar{S_2}$ due to the way $G_{w^{\epsilon}, S}$ is constructed.
  Similarly, if $v_{i} \in \bar{S_1} \cap \bar{S_2}$, then we must
  have $v_{i+1} \in S_1 \cap S_2$.
  Let $P_{\text{mid}} = \{v_2, \ldots, v_{k-1}\}$, $I = S_1 \cap S_2$, and
  $O = \bar{S_1} \cap \bar{S_2}$.
  Clearly, we have
  \begin{equation}
    |\hat{S}_1 \cap \hat{S}_2| - |S_1 \cap S_2|
      = |P_{\text{mid}} \cap O| - |P_{\text{mid}} \cap I| +
      \llbracket v_1 \in \hat{S}_1 \cap \hat{S}_2 \rrbracket +
      \llbracket v_k \in \hat{S}_1 \cap \hat{S}_2 \rrbracket.
    \label{eq:diff}
  \end{equation}
  We prove the claim by considering the following four possible cases.
  \begin{description}
    \item[{$k$ is even and $v_2 \in I$}.] We have
      $(v_1, v_2) \in E_2$, $(v_{k-1}, v_k) \in E_2$, and
      $|P_{\text{mid}} \cap I| = |P_\text{mid} \cap O|$.
      Also, $v_1 \in \tail(P \cap E_2)$ and therefore
      $v_1 \in \hat{S}_1 \cap \hat{S}_2$.
    \item[{$k$ is even and $v_2 \in O$}.] We have $(v_1, v_2) \in E_1$,
      $(v_{k-1}, v_k) \in E_1$, and
      $|P_{\text{mid}} \cap I| = |P_\text{mid} \cap O|$.
      Also, $v_k \in \operatorname{head}(P \cap E_1)$ and therefore
      $v_k \in \hat{S}_1 \cap \hat{S}_2$.
    \item[{$k$ is odd and $v_2 \in I$}.] We have $(v_1, v_2) \in E_2$,
      $(v_{k-1}, v_k) \in E_1$, and
      $|P_{\text{mid}} \cap I| = |P_\text{mid} \cap O| + 1$.
      Also,
      $v_1 \in \tail(P \cap E_2)$, $v_k \in \head(P \cap E_1)$
      and therefore $v_1, v_k \in \hat{S}_1 \cap \hat{S}_2$.
    \item[{$k$ is odd and $v_2 \in O$}.] We have $(v_1, v_2) \in E_1$,
      $(v_{k-1}, v_k) \in E_2$, and
      $|P_{\text{mid}} \cap I| = |P_\text{mid} \cap O| - 1$.
  \end{description}
  In all cases, we have $|\hat{S}_1 \cap \hat{S}_2| > |S_1 \cap S_2|$ via
  \Cref{eq:diff}, concluding the proof.
\end{proof}

We prove the following claims for $i = 1$.
The proofs for $i = 2$ follow analogously.

\begin{claim}
  $\hat{S}_i \in \mathcal{I}_i$ holds for each $i \in \{1, 2\}$.
  \label{claim:keep_independence}
\end{claim}

\begin{proof}
  Let $P_1 = P \cap E_1 = \{(b_1, a_1), (b_2, a_2), \ldots (b_p, a_p)\}$, where
  $(S_1 \setminus \{b_i\}) \cup \{a_i\} \in \mathcal{I}_1$ holds for each
  $1 \leq i \leq p$.
  Since $P$ is the shortest path,
  \[
    d(b_i) + u_1(b_i) - u_1(a_i) = d(a_i)
      \implies d(b_i) + u_1(b_i) = d(a_i) + u_1(a_i)
  \]
  holds for each $i$.
  Reorder $P_1$ so that
  $d(b_1) + u_1(b_1) \leq d(b_2) + u_1(b_2) \leq \cdots \leq d(b_p) + u_1(b_p)$.
  Moreover, if $d(b_i) + u_1(b_i) = d(b_j) + u_1(b_j)$ for some $i, j$, then
  $(b_i, a_i)$ precedes $(b_j, a_j)$ in $P_1$ if and only if $(b_i, a_i)$
  precedes $(b_j, a_j)$ in $P$.
  It follows that for each $1 \leq j < i \leq p$, it holds that
  $(S_1 \setminus \{b_j\}) \cup \{a_i\} \not\in \mathcal{I}_1$ since otherwise
  we would have
  \begin{equation}
    d(b_j) + u_1(b_j) - u_1(a_i) \geq d(a_i)
      \implies d(b_j) + u_1(b_j) \geq d(a_i) + u_1(a_i).
    \label{eq:contradiction}
  \end{equation}
  Because $j < i$, \eqref{eq:contradiction} must take equality, but this would
  contradict with the fact the $P$ has the least number of edges since the edge
  $(b_j, a_i)$ ``jumps'' over vertices $a_j, b_{j+1},\ldots,b_i$ in $P$ and has
  the same weight as the subpath $b_j, a_j, \ldots, b_i, a_i$.
  As such, by \Cref{lemma:can_augment}, the claim is proved.
\end{proof}

\begin{claim}
  $\hat{S}_i$ is $\hat{w}^{\epsilon}_i$-maximum for each $i \in \{1, 2\}$.
  \label{claim:keep_maximum}
\end{claim}

\begin{proof}
  Let $P_1 = P \cap E_1 = \{(b_1, a_1), \ldots, (b_p, a_p)\}$ be ordered the
  same way as in the proof of \Cref{claim:keep_independence}.
  It suffices to show that
  $\hat{w}^{\epsilon}_1(x) \geq \hat{w}^{\epsilon}_1(y)$ holds for each
  $x \in \hat{S}_1$ and $y \not\in \hat{S}_1$ with
  $(\hat{S}_1 \setminus \{x\}) \cup \{y\} \in \mathcal{I}_1$.
  Consider the following two cases.
  
  \begin{enumerate}
    \item $(S_1 \setminus \{x\}) \cup \{y\} \in \mathcal{I}_1$: Since
      $(x, y) \in E_1$, it follows that
      \[
        d(x) + u_1(x) - u_1(y) \geq d(y)
          \implies \hat{w}^{\epsilon}_1(x) = d(x) + u_1(x) \geq
                                   d(y) + u_1(y) = \hat{w}^{\epsilon}_1(y).
      \]
    \item $(S_1 \setminus \{x\}) \cup \{y\} \not\in \mathcal{I}_1$: By
      \Cref{lemma:new_edges}, there exists $1 \leq \ell \leq k \leq p$ such that
      (1) $(S_1 \setminus \{x\}) \cup \{a_k\} \in \mathcal{I}_1$ and either
      (2.1) $b_{\ell} = y$ or
      (2.2) $(S_1 \setminus \{b_{\ell}\}) \cup \{y\} \in \mathcal{I}_1$.
      (1) implies that
      $\hat{w}^{\epsilon}_1(x) \geq \hat{w}^{\epsilon}_1(a_k)$.
      If (2.1) holds, then
      $\hat{w}^{\epsilon}_1(x) \geq \hat{w}^{\epsilon}_1(a_k)
        \geq \hat{w}^{\epsilon}_1(b_{\ell}) = \hat{w}^{\epsilon}_1(y)$.
      If (2.2) holds, then
      $\hat{w}^{\epsilon}_1(x) \geq \hat{w}^{\epsilon}_1(a_k)
        \geq \hat{w}^{\epsilon}_1(b_{\ell}) \geq \hat{w}^{\epsilon}_1(y)$.
  \end{enumerate}
  
  The claim is proved.
\end{proof}

\end{document}